\documentclass[a4paper,11pt,twoside]{amsart}

\usepackage[T1]{fontenc}
\usepackage{amsfonts,amsmath,amssymb,amsxtra,amscd}
\usepackage{mathrsfs}
\usepackage[Symbol]{upgreek}
\usepackage[nice]{nicefrac}
\usepackage{cancel}
\usepackage{ae}
\usepackage{dsfont}
\usepackage[all]{xy}
\usepackage{graphicx}
\usepackage{layout}

\def\plus#1#2{\vrule height#1pt width0pt depth#2pt}
\def\@{\hskip.8pt}
\def\?{\hskip.3pt}

\theoremstyle{plain}
\newtheorem{theorem}{Theorem}[section]

\newtheorem{definition}{Definition}[section]
\newtheorem{proposition}{Proposition}[section]

\newtheorem{remark}{Remark}[section]

\numberwithin{equation}{section}

\newcommand{\rarw}[1]{\overset{#1}{\longrightarrow}}

\newcommand{\de}{\partial}

\def\A{\mathcal{A}}
\def\B{\mathcal{B}}
\def\AA{\hat{\mathcal{A}}}
\def\C{\mathcal{C}}
\def\CC{\mathscr{C}}

\newcommand{\D}[1]{\frac{D #1}{D t}}

\def\F{\mathscr{F}}
\def\g{\gamma}

\def\GAMMA{\hat{\gamma}}

\def\H{\mathcal{H}}
\def\Ham{\mathscr{H}}
\def\I{\mathcal{I}}
\def\L{\mathcal{L}}
\def\Lagr{\mathscr{L}}

\def\Q{\dot{q}}
\def\R{\mathds{R}}
\def\SS{\mathcal{S}}
\def\tgamma{{\tilde{\gamma}}}
\def\ThPPC{\Theta_{\text{\tiny PPC}}}
\def\u{\upsilon}
\def\U{\dot{u}}
\def\V{\mathcal{V}_{n+1}}
\def\Vg{V\?(\gamma)}
\def\Ag{A\@(\GAMMA)}

\def\d#1/d#2{\frac{d\/#1}{d\/#2}}
\def\de#1/de#2{\frac{\partial\/#1}{\partial\/#2}}
\def\SD#1/de#2/de#3{\ifx#2 \frac{\plus02\partial^{\@\@2}#1}
    {\plus90\partial\@#3^{\@2}} \else\frac{\plus02\partial^{\@\@2}#1}
    {\partial\?#2\partial\?#3}\fi}
\def\TD#1/de#2/de#3/de#4{\ifx#2 \frac{\plus02\partial^{\@\@3}#1}%
    {\partial\?{#3}^2\partial\?#4} \else\frac{\plus02\partial^{\@\@3}#1}%
    {\partial\?#2\partial\?#3\partial\?#4}\fi}
\def\D#1/D#2{\frac{D\/#1}{D\/#2}}
\def\DD#1/D#2{\textstyle{\text{\Large$\D{#1}/D{#2}$}}}
\def\dd#1/d#2{\textstyle{\text{\Large$\d{#1}/d{#2}$}}}
\def\De#1/de#2{\textstyle{\text{\Large$\de{#1}/de{#2}$}}}
\def\sD#1/de#2/de#3{\textstyle{\text{\Large$\SD{#1}/de{#2}/de{#3}$}}}
\def\tD#1/de#2/de#3/de#4{\textstyle{\text{\Large$\TD{#1}/de{#2}/de{#3}/de{#4}$}}}

\catcode`\@=11
\renewenvironment{subequations}{%
  \refstepcounter{equation}%
  \protected@edef\theparentequation{\theequation}%
  \setcounter{parentequation}{\value{equation}}%
  \setcounter{equation}{0}%
  \def\theequation{\theparentequation\hspace{1pt}\alph{equation}}%
  \ignorespaces
}{%
  \setcounter{equation}{\value{parentequation}}%
  \ignorespacesafterend
} \catcode`\@=12

\def\Ref#1#2{\if#2)\ref{#1}#2\else\ref{#1}\@#2\fi}

\def\inizio#1{\ifx)#1)\endzag\else\hskip.1pt,\let\zag=\Zag\fi\zag}
\def\zig#1{\ifx)#1)\endzag\else\hskip.1pt,\fi\zag}
\def\zag{\relax}
\def\Zag#1{\@#1\zig}
\def\endzag{\let\zag=\relax}

\def\eref#1#2{\ref{#1}\ifx,#2\hspace{.1pt},\let\inizio=\Zag\let\zag=\Zag\else\@#2\fi\inizio}

\begin{document}


\markboth{E. Pagani \textit{et al.} } {On the gauge structure of the
calculus of variations with constraints}

\title{ON THE GAUGE STRUCTURE OF THE CALCULUS OF 
VARIATIONS WITH CONSTRAINTS}

\author{D. Bruno\?,\, G. Luria\?,\, E. Pagani }
\address{Dipartimento di Matematica dell'Universit\`a di Genova \\
        Via Dodecaneso, 35 - 16146 Genova (Italia)}
\email{bruno@dima.unige.it}
\address{DIPTEM Sez. Metodi e Modelli Matematici\@, Universit\`a di Genova\\
         Piazzale Kennedy, Pad. D - 16129 Genova (Italia)}
         \email{luria@diptem.unige.it}
\address{Dipartimento di Matematica dell'Universit\`a di Trento \\
        Via Sommarive, 14 - 38050 Povo di Trento (Italia)}
\email{pagani@science.unitn.it}

\begin{abstract}
A gauge--invariant formulation of constrained variational calculus,
based on the introduction of the bundle of \emph{affine scalars}
over the configuration manifold, is presented. In the
resulting setup, the ``\?Lagrangian\/'' $\Lagr\?$ is replaced by a
section of a suitable principal fibre bundle over the velocity
space. A geometric rephrasement of Pontryagin's maximum
principle, showing the equivalence between a constrained variational
problem in the state space and a canonically associated 
free one in a higher affine bundle,
is proved.\\

\noindent
{\footnotesize \textit{Keywords:} Constrained calculus of variations; Optimal control
theory; Gauge theory; First variation.}

\vspace{2pt}\noindent
{\footnotesize Mathematical Subject Classification 2010: 49J, 70Q05, 37J60,
       70H03, 53B05, 53D10}

\end{abstract}

\maketitle

%
%
%
%
%
%
%

\section*{Introduction}
Almost ten years ago, a mathematical setting for a
formulation of Classical Mechanics, automatically embodying its gauge
invariance, has been introduced \cite{MPL,MVB}. Besides that, a more
recent paper \cite{MBP} proposes a geometric revisitation
of the calculus of variations in the presence of
non--holonomic constraints.

The present work uses the arguments of
\cite{MBP} and the geometrical framework provided by
\cite{MPL,MVB} to analyze the underlying gauge structure of constrained variational calculus.
As a by--product, the resulting scheme allows to obtain the relevant ingredients that are 
commonly used in the variational context as \emph{canonical} geometrical objects.
 
The topic will be developed
within the family of \emph{differentiable} curves alone, thus avoiding
all the issues coming from the possible presence of corners which are 
of poor significance from the gauge--theoretical point of view. 
The extension to the piecewise differentiable case, surely 
interesting on its own, can however be easily pursued.

\smallskip
Consider an abstract system $\@\B\?$, subject to a set of differentiable 
and possibly non--holonomic constraints, and let $\?\A\@$ denote 
its admissible velocity space. 

Then, define an \emph{action functional} by integrating
a suitable differentiable ``\@cost function\?'' on 
$\@\A\@$ --- called the \emph{Lagrangian} --- along 
the admissible evolutions of the system.
In its essential features, the variational problem we shall deal with is the one of
characterizing, among all those evolutions of
$\?\B\@$ which fulfil the restrictions imposed by  
the assigned constraints and connect a fixed pair of
configurations, the extremals (\?if at all\?) of the given action
functional. 
As we shall see, the whole topic has 
very close links with optimal control theory. 

\smallskip
For the sake of convenience, the first part of
the paper shall be devoted as a reference tool consisting of a brief review
of a few basic aspects of jet--bundle geometry and 
non--holonomic geometry 
as well as of
those contents of \cite{MBP,MPL,MVB} that will be involved in 
the subsequent discussion. 
This will include, among other things, a revisitation of 
the Lagrangian and Hamiltonian bundles.
All results will be stated without any proofs nor comments.

\smallskip
Afterwards, we shall develop an algorithm able to establish a 
canonical correspondence between
the input data of the problem, namely the constraints and
the  Lagrangian function, and an alternative
\emph{free} variational problem over a distinguished overlying affine bundle
$\@\CC\/(\A)\to \A\@$. In the resulting setup, under
suitable hypotheses, the gauge--independent problem
in $\@\CC\/(\A)\@$ is proved to be equivalent to the actual
constrained one. This clarifies the geometrical
essence of Pontryagin's method, based on the
introduction of the \emph{costates}.

\smallskip
As the final element, the circumstances under which the
correspondence between the solutions of the variational 
problem in $\@\A\@$ and those of its associated problem in $\@\CC\/(A)\@$
is $\?1-1\@$ are investigated.


\vspace{2pt}
\section{Geometric setup}
\subsection{Preliminaries}
Let $\@\V\rarw{t} \R\@$ denote an $\@(n + 1)$--dimensional fibre
bundle, henceforth called the \emph{event space} and referred to
local fibred coordinates $\@t, q^1, \ldots, q^n\@$.

Every section $\@\gamma \colon \R \to \V\@$, locally described as
$\@q^i = q^i(t)\@$, will be interpreted as an \emph{evolution} of an
abstract system $\@\B\@$, parameterized in terms of the independent
variable $\@t\@$. The first jet--space $\@j_1(\V) \rarw{\pi} \V\@$
is an affine bundle over $\@\V\@$, modelled on the vertical space
$\?V(\V)\@$.\linebreak By the very definition of jet--bundle, every section
$\@\gamma \colon \R \to \V\@$ may be lifted to a section
$\@j_1(\gamma) \colon \R \to j_1(\V)\@$, simply by assigning to each
$\@t \in \R\@$ the tangent vector to $\gamma\@$. The section
$\@j_1(\gamma)\@$ will be called the \emph{jet--extension} of
$\@\gamma\@$ on $\@j_1(\V)\@$\vspace{1pt}.

\vspace{1pt}
Both spaces $\@j_1(\V)\@$ and $\@V(\V)\@$ may be viewed as
submanifolds of the tangent space $\@T(\V)\@$ according to the
identifications
\begin{subequations}
\begin{align}
j_1(\V) &= \left\{\@ z \in T(\V)\ \vline\ \left<\?z\?,\?dt\?\right>
= 1 \@\right\}
\label{Eq_1.1a}\\[2pt]
V(\V) &= \left\{ {\textrm{v}} \in T(\V)\ \vline\
\left<\?{\textrm{v}}\?,\?dt\?\right> = 0 \@\right\} \label{Eq_1.1b}
\end{align}
\end{subequations}

\vspace{2pt}
The terminology is borrowed from Classical Mechanics, where 
$\@\B\@$ is identified with a material system,
the manifold $\@\V\,$ with its configuration space--time, the
projection $\/t:\V\to\R\@$\vspace{1pt} with the \emph{absolute
time\/} function and the jet--space $\@j_1\/(\V)\@$ with the
\emph{velocity space} of $\@\B\@$.

\smallskip The geometry of the manifold $j_1\/(\V)$ will be regarded
as known. The reader is referred to \cite{MBP,Saunders} for the notation, the
terminology and a thorough analysis. Unless otherwise
stated, given any set of local coordinates on $\?\V\@$, the 
corresponding local jet--coordinate system on
$j_1\/(\V)$ will be denoted by $t, q^1, \ldots, q^n, \Q^1, \ldots,
\Q^n$.

\bigskip \noindent The dual of the vertical bundle, henceforth
denoted by $V^*\/(\V) \rarw{\pi} \V$, is, in view of
eq.\@\eqref{Eq_1.1b}, canonically isomorphic to the quotient of the
cotangent space $T^*\/(\V)$ by the equivalence relation
\begin{equation*}
  \sigma \sim \sigma'\quad \Longleftrightarrow \quad \left\{\begin{array}{ll}
  \pi(\sigma) = \pi(\sigma')\\[0.3em] \sigma - \sigma' \propto dt_{\, \vline\, \pi(\sigma)}
  \end{array}\right.
\end{equation*}
Every local coordinate system $\@t, q^i\@$ in $\@\V\@$ induces
fibred coordinates $\@t, q^i, \hat{p}_i\@$ in $\@V^*(\V)\@$, with
\[
\hat{p}_i(\hat{\sigma})\, :=\, \left<\?\hat{\sigma}\? ,\? \left(\de /de
{q^i}\right)_{\pi(\hat{\sigma})}\?\right> \qquad \forall\
\hat{\sigma} \in V^*(\V)
\]
and transformation laws
\begin{equation}\label{Eq_1.2}
  \bar{t} = t + c\ ,\qquad \bar{q}\@^i = \bar{q}\@^i(t, q^1, \ldots, q^n)\ ,
\qquad \bar{\hat{p}}_i = \hat{p}_k\, \de {q^k} /de {\bar{q}\@^i}
\end{equation}

\medskip The annihilator of the tangent distribution to the totality
of the jet--extensions of sections $\@\gamma\@$ is a subspace
$\@\C\/\big(j_1\/(\V)\big)\@$ of $\@T^*\/\big(j_1\/(\V)\big)\@$,
called the \emph{contact bundle}.

\smallskip \noindent Alternatively, this last may be seen as the pull--back of
the space $V^*\/(\V)$ through the map $j_1\/(\V) \rarw{\pi} \V$.
As such, $\@\C\/\big(j_1\/(\V)\big)\@$ is, at the same time, a vector
bundle over $\?j_1\/(\V)\@$ and an affine bundle
over $\?V^*\/(\V)\@$. The manifold
$\@\C\/\big(j_1\/(\V)\big)\@$ will be referred to coordinates $t,
q^i, \Q^i, \hat{p}_i$, related in an obvious way to those
in $\@j_1\/(\V)\,$ and in $\@V^*\/(\V)\,$. 
Every $\@\sigma \in \C\/\big(j_1\/(\V)\big)\@$ will
be called a \emph{contact $1$--form} over $\?j_1\/(\V)\@$.


\subsection{Non--holonomic constraints}\label{SSEc_Non-holonomic constraints}
Let $\@\A\@$ denote an embedded submanifold of $\@j_1(\V)\@$, fibred
over $\@\V\@$. The situation, summarized into the following commutative
diagram
\begin{equation}\label{Diag_1.4}
  \begin{CD}
    \A @>i>> j_1(\V) \\ @V{\pi}VV @VV{\pi}V \\ \V @= \V
  \end{CD}
\end{equation}
provides the natural setting for the study of non--holonomic
constraints.

\smallskip
 The manifold $\@\A\@$ is referred to local fibred coordinates $\@t, q^1,
\ldots, q^n, z^1, \ldots, z^r\@$ with transformation laws
\begin{equation}\label{Eq_1.5}
\bar{t} = t + c\,,\quad\;\bar{q}\@^i=\bar{q}\@^i\/(t,q^1,\ldots,q^n)\,,
\quad\;\bar{z}\@^A=\bar{z}\@^A\/(t,q^1,\ldots,q^n,z^1,\ldots,z^r)
\end{equation}
while the imbedding $\@i: \A \to j_1(\V)\@$ is locally expressed as
\begin{subequations}\label{Eq_1.6}
\begin{equation}\label{Eq_1.6a}
  \Q^i = \psi^i(t, q^1, \ldots, q^n, z^1, \ldots, z^r)\ ,\quad i = 1, \ldots, n
  \quad ,\quad \text{rank}\, \bigg\| \de {\left(\psi^1\, \cdots\,
\psi^n\right)} /de {\left(z^1\, \cdots\, z^r\right)} \bigg\| = r
\end{equation}
or, alternatively, may be implicit represented as
\begin{equation}\label{Eq_1.6b}
    g^\sigma \left(t, q^1, \ldots, q^n, \Q^1, \ldots, \Q^n\right) = 0\ , \ \ \sigma = 1, 
    \ldots, n - r\ ,\ \ \text{rank}\, \Big\| \de {\left(g^1\, \cdots\,
g^{n-r}\right)} /de {\left(\Q^1\, \cdots\, \Q^n\right)} \Big\| = n -r
\end{equation}
\end{subequations}
In the following, we shall not distinguish
between the manifold $\@\A\@$ and its image $\@i(\A) \subset
j_1(\V)\@$.

\medskip
In the presence of non--holonomic constraints, an evolution
$\@\gamma \colon \R \to \V\@$ is called
\emph{admissible} if and only if
its first jet--extension is contained in $\@\A\@$, namely if there
exists a section $\@\GAMMA \colon \R \to \A\@$ satisfying
$\@j_1(\pi\@\cdot\@\GAMMA) = i\@\cdot\@\GAMMA\@$. 
Expressing any section $\@\GAMMA\@$ in coordinates as
$\@q^i =q^i\/(t)\,,\,z^A=z^A\/(t)\,$, the admissibility requirement
takes the explicit form
\begin{equation}\label{Eq_1.7}
\d {q^i} /d t \,=\, \psi\?^i\/(t,q^1\/(t),\ldots,q^n\/(t),z^1\/(t),\ldots,z^r\/(t)\@)
\end{equation}

\medskip \noindent The concepts of vertical vector and contact
$1$--form are easily extended to the submanifold $\A\@$: as usual,
the vertical bundle $\@V(\A)\@$ is the kernel of the push--forward
$\@\pi_*\colon T(\A) \to T(\V)\@$ while the contact bundle
$\@\C(\A)\@$ is the pull--back on $\A$ of the bundle
$\@\C(j_1(\V))\@$, as expressed by the commutative diagram
\begin{equation}\label{Diag_1.8}
  \begin{CD}
    \C(\A) @>>> \C(j_1(\V)) @>>> V^*(\V) \\ @VVV @VVV @VVV \\
    \A @>i>> j_1(\V) @>\pi>> \V
  \end{CD}
\end{equation}
The latter allows to regard the contact bundle $\@\C(\A)\@$ as a
fibre bundle over the space $\@V^*(\V)\@$, identical to the
pull-back of $\@V^*(\V)\@$ through the map 
$\@\A \to \V\@$.

\subsection{Infinitesimal deformations of sections}
\label{SSec_Infinitesimal deformations of sections}
\textbf{(\/i\/)} \,Quite generally, given a section $\@\g:\R\to \V$,\linebreak
a (\@weak\@) \emph{deformation\/} of $\@\g\@$ is a 1--parameter
family of sections
$\@\g\?_\xi\,,\;\xi\in(-\varepsilon,\varepsilon)\,$ depending
differentiably on $\/\xi\@$ and satisfying $\@\g\?_0=\g\@$.

In the presence of non--holonomic constraints,
a deformation
$\@\gamma\?_\xi\,$ is called \emph{admissible\/} if and only if each
section $\@\gamma\?_\xi:\R\to\V\,$ is admissible in the sense of
\S\@\ref{SSEc_Non-holonomic constraints}. In a similar way,
a deformation $\@\GAMMA\?_\xi\@$ of an admissible section
$\@\GAMMA:\R\to\A\,$ is called admissible if and only if all
sections $\@\GAMMA\?_\xi:\R\to\A\,$ are admissible.

By definition, the
admissible sections $\@\gamma:\R\to\V\,$ are in 1--1 correspondence
with the admissible sections $\@\GAMMA:\R\to\A\,$ through the
relations
\begin{equation}\label{Eq_gamma<=>GAMMA}
\gamma=\pi\cdot\GAMMA\;,\qquad\quad j_1(\gamma)=i\cdot\GAMMA
\end{equation}
Every admissible deformation of $\@\gamma\@$ may therefore be
expressed as $\@\gamma\?_\xi=\pi\cdot\GAMMA\?_\xi\@$, being
$\@\GAMMA\?_\xi:\R\to\A\,$ an admissible deformation of
$\@\GAMMA\@$.

In coordinates, preserving the representation
$\@\GAMMA:\@q^i=q^i\/(t)\,,\,z^A=z^A\/(t)\,$, the admissible
deformations of $\?\GAMMA\@$ are described by equations of the form
\begin{equation}
\GAMMA\?_\xi\@:\qquad q^i=\varphi^i\/(\xi,t)\,,\quad z^A=\zeta^A\/(\xi,t)
\end{equation}
subject to the conditions
\begin{subequations}
\begin{align}
 & \varphi^i\/(0,t)\@=\@q^i\/(t)\,,\quad \zeta^A\/(0,t)\@=\@z^A\/(t) \\[4pt]
 & \de\varphi^i/de t\,=\,\psi^i\/(t,\varphi^i\/(\xi,t),\zeta^A\/(\xi,t))
 \label{Eq_Ammissibilit GAMMA_xi}
\end{align}
\end{subequations}

\medskip
For each $\/t\in \R\?$, the curve $\@\xi\to\GAMMA\?_\xi\/(t)\,$ is
called the \emph{orbit\/} of the deformation $\@\GAMMA\?_\xi\@$ through
the point $\?\GAMMA\/(t)\,$. The vector field along $\@\GAMMA\@$ tangent
to the orbits at $\?\xi=0\,$ is called the \emph{infinitesimal
deformation\/} associated with $\@\GAMMA\?_\xi\@$.

Setting
\begin{equation}\label{Eq_X^i,Gamma^A}
X^i\/(t):=\left(\de\@\varphi^i/de{\?\xi\;}\right)_{\!\xi=0}\;\hskip4.32pt,\qquad\  \Gamma^A\/(t):=
\left(\de\@\zeta\/^A/de{\?\xi\;}\right)_{\!\xi=0}
\end{equation}
the infinitesimal deformation tangent to $\@\GAMMA_\xi\@$ is
described by the vector field
\begin{equation}\label{Eq_Deformaz infinitesima GAMMA}
\hat{X}\,=\,
X^i\/(t)\left(\de/de{q^i}\right)_{\!\GAMMA}\,+\,\Gamma^A\/(t)\left(\de/de{z^A}\right)_{\!\GAMMA}
\end{equation}
while equation  \eqref{Eq_Ammissibilit GAMMA_xi} is reflected into
the relation
\begin{equation}
\d X^i/d t\,=\,\left.\de /de {\?t\;}\@ \de {\@\varphi^i} /de {\?\xi\;}\@\right|_{\xi=0}\,=
\,\left(\de\@\psi^i/de{q^k}\right)_{\!\GAMMA}X^k\,
+\,\left(\de\@\psi^i/de{z^A}\right)_{\!\GAMMA}\Gamma^A\\[5pt]
\end{equation}
commonly referred to as the \emph{variational
equation\/}.

\noindent The infinitesimal deformation tangent to the projection
$\@\gamma\?_\xi=\pi\cdot\GAMMA\?_\xi\,$ is similarly described by
the field
\begin{equation}
X=\pi_*\@\hat{X}\,=\,
\left(\de\@\varphi^i/de{\?\xi\;\,}\right)_{\!\xi=0}\,\left(\de/de{q^i}\right)_{\!\gamma}=\,
X^i\/(t)\;\left(\de/de{q^i}\right)_{\!\gamma}
\end{equation}

\bigskip
\noindent \textbf{(\/ii\/)} \,We now want to sketch out the
construction of a suitable geometrical environment for the
description of such infinitesimal
deformations.
In order to do so, we denote by
$\@\Vg\xrightarrow{t\,}\R\@$ the vector bundle over $\@\R\@$ formed
by the totality of vertical vectors along $\@\gamma\@$, and by
$\@\Ag\xrightarrow{t\,}\R\@$\vspace{1pt} the analogous bundle formed
by the totality of vectors along $\@\GAMMA\@$ annihilating the
$1$--form $d\/t\?$. Both bundles  are referred to fibred
coordinates --- the former to $\?t,v^i\@$ and the latter to
$\?t,v^i\!,w^A\@$ --- according to the prescriptions
\begin{alignat*}{2}
  & X\in\Vg\quad\Longleftrightarrow\quad &&
  X\,=\,v^i\/(X)\?\bigg(\de /de {q^i}\bigg)_{\!\gamma\/(t\/(X))}
  \\[5pt]
  & \hat{X}\in\Ag\quad\Longleftrightarrow\quad &&
  \hat{X}\,=\,v^i\/(\hat{X})\?\bigg(\de /de {q^i}\bigg)_{\!\GAMMA\/(t\/(\hat{X}))}+\;
  w^A\/(\hat{X})\?\bigg(\de /de {z^A}\bigg)_{\!\GAMMA\/(t\/(\hat{X}))}\quad
\end{alignat*}

As proved in \cite{MBP}\?, the
first jet--bundle $\@j_1(\Vg)\@$ is canonically isomorphic to the
space of vectors along the jet--extension $\@j_1(\gamma)\@$ annihilating the $1$--form
$\?dt\/$.

In jet--coordinates, the identification is expressed by
the relation
\begin{equation*}
Z\in j_1(\Vg)\quad \Longleftrightarrow \quad Z\,=\,v^i\/(Z)\?\bigg(\de /de {q^i}\bigg)_{\!j_1(\gamma)\/(t\/(Z))} +
\;\dot{v}^i\/(Z)\?\bigg(\de /de {\Q^i}\bigg)_{\!j_1(\gamma)\/(t\/(Z))}
\end{equation*}
The push--forward of the imbedding $\@\A\xrightarrow{i\,}j_1(\V)\@$,
restricted to the subspace $\@\Ag\subset T\?(\A)\@$, makes the
latter into a subbundle of $\@j_1({\Vg})\@$. This gives rise to a
fibred morphism
\begin{subequations}
\begin{equation}\label{Diag_1.12a}
\begin{CD}
\Ag          @>{i_*}>>       j_1(\Vg)      \\
@V{\pi_*}VV                @VV{\pi_*}V     \\
\Vg\;\;      @=             \Vg
\end{CD}
\end{equation}
expressed in coordinates as
\begin{equation}\label{Eq_1.12b}
\dot{v}^i=\bigg(\de {\psi^i} /de {q^k}\bigg)_{\!\GAMMA}\@v^k +
\bigg(\de {\psi^i} /de {z^A}\bigg)_{\!\GAMMA}\@w^A
\end{equation}
\end{subequations}

\bigskip
\noindent All previous results are then summarized into the following

\begin{proposition}\label{Prop_Infinitesimal deformations}
Let $\@\gamma:\R\to\V\,$ and $\@\GAMMA:\R\to\A\,$ denote two
admissible sections, related by equation (\ref{Eq_gamma<=>GAMMA}).
Then:
\begin{enumerate}
\item[i)]
the infinitesimal deformations of $\?\gamma\?$ and those of
$\?\GAMMA\?$ are respectively expressed as sections
$\@X:\R\to\Vg\@$ and $\@\hat{X}:\R\to\Ag\@$;
\item[ii)]
a section $\?X:\R\to\Vg\?$ represents an \emph{admissible\/}
infinitesimal deformation of $\@\gamma\@$ if and only if its
first jet--extension factors through $\@\Ag\@$, i.e.~if and only
if there exists a section $\@\hat{X}:\R\to\Ag\,$ satisfying
$\@j_1(X)=i_*\@\hat{X}\/$; conversely,~a~section
$\@\hat{X}:\R\to\Ag\@$ represents an admissible infinitesimal
deformation of $\@\GAMMA\@$ if and only if it projects into an
admissible infinitesimal deformation~of $\@\gamma\@$, i.e.~if
and only if $\,i_*\@\hat{X}= j_1({\pi_*\@\hat{X}})\@$.
\end{enumerate}
\end{proposition}

From a structural viewpoint, Proposition \ref{Prop_Infinitesimal
deformations}  points out the perfectly symmetric roles
respectively played by diagram \eqref{Diag_1.4} in the study of the
admissible \emph{evolutions\/} and by diagram \eqref{Diag_1.12a} in
the study of the infinitesimal \emph{deformations\/}, thus enforcing
the idea that the second context is essentially a ``\@linearized
counterpart\/'' of the former one.


\subsection{The gauge setup}
\subsubsection{The Lagrangian bundles}\label{SSSec_The Lagrangian bundles}
Given any system subject to (smooth) positional constraints, we
introduce a double fibration $\@P \rarw{\pi} \V \rarw{t} \R\@$,
where:
\begin{itemize}
  \item[\textit{i})] $\V \rarw{t} \R\@$ is the configuration space--time of the
  system;\vspace{4pt}
  \item[\textit{ii})] $P \rarw{\pi} \V\@$ is a principal fibre bundle with
  structural group $\@(\?\R\?,\? +\?)\,$.
\end{itemize}

As a consequence of the stated definition, each fibre $\@P_x :=
\pi^{-1}(x),\, x \in \V\@$ is an affine\linebreak 1--space. The total space
$\@P\@$ is therefore a trivial bundle, diffeomorphic in a
non--canonical way to the Cartesian product $\@\V \times \R\@$,
called the bundle of \emph{affine scalars} over $\@\V\@$.

The action of $\@(\?\R\?,\? +\?)\@$ on $\@P\@$ results into a
$1$--parameter group of diffeomorphisms $\@\psi_\xi\colon P \to
P\@$, conventionally expressed through the additive notation
\begin{equation}\label{Eq_1.14}
  \psi_\xi\/(\nu)\, :=\, \nu + \xi \qquad \forall\, \xi \in \R,\ \nu \in P
\end{equation}
Every map $\@u \colon P \to \R\@$ satisfying the requirement
\[
u\/(\nu + \xi)\, =\, u(\nu) + \xi
\]
is called a (\@global\@) trivialization of $\@P\@$. If $\@u\@$,
$u'\@$ is any pair of trivializations, the difference $\@u - u'\@$
is then (\@the pull--back of\@) a function over $\@\V\@$. Moreover,
every section $\@\varsigma \colon \V \to P\@$ determines a
trivialization $\@u_\varsigma \in \F(P)\@$ and conversely, being the
relation between $\@\varsigma\@$ and $\@u_\varsigma\?$ 
expressed by the condition
\begin{equation}\label{Eq_1.15}
  \nu\, =\, \varsigma(\pi(\nu))\, +\, u_\varsigma(\nu)\qquad \forall\, \nu \in P
\end{equation}
Therefore, once a (\@global\@) trivialization $\@u \colon P \to
\R\@$ has been chosen, every section $\@\varsigma\colon  \V \to
P\@$\vspace{.8pt} is completely characterized by the knowledge of
the function $\@f = \varsigma^*(u) \in \F(\V)\@$.

The assignment of $\@u\@$ allows to lift every local coordinate
system $\@t, q^1, \ldots, q^n\@$ over $\@\V\@$ to a corresponding
fibred one $\@t, q^1, \ldots, q^n, u\@$ over $\@P\@$, being the most
general transformation between fibred coordinates of the form
\[
\bar{t} = t + c\ ,\qquad \bar{q}\@^i = \bar{q}\@^i(t, q^1, \ldots, q^n)\ ,
\qquad \bar{u} = u + f(t, q^1. \ldots, q^n)
\]
The action of the group $\@(\?\R\?,\? +\?)\@$ on the manifold
$\@P\@$ is expressed in fibred coordinates by the relations
\[
t(\nu + \xi) = t(\nu)\ ,\qquad q^i(\nu + \xi) = q^i(\nu)\ ,
\qquad u(\nu + \xi) = u(\nu) + \xi
\]
As a result, the generator of the group action \eqref{Eq_1.14},
usually referred to as the \emph{fundamental vector field} of
$\?P\?$, is canonically identified with the field
$\@\tfrac{\partial}{\partial u}\@$.

\vspace{1pt}The (\@pull--back of the\@) absolute time function
determines a fibration $\@P \rarw{t} \R\@$ whose associated first
jet--space is indicated by $\@j_1\/(P,\? \R) \rarw{\pi} P\@$
and is referred to local jet--coordinates $\@t, q^i,
u, \Q^i, \dot{u}\@$ subject to transformation laws
\begin{subequations}\label{Eq_1.16}
\begin{equation}
\bar{t} = t + c\ ,\qquad \bar{q}\@^i = \bar{q}\@^i(t, q^1, \ldots, q^n)\ ,
\qquad \bar{u} = u + f(t, q^1. \ldots, q^n)
\end{equation}
\begin{equation}
  \bar{\Q}\@^i = \frac{\partial \bar{q}\@^i}{\partial q^k}\, \Q^k + \frac{\partial \bar{q}\@^i}{\partial t}\ , \qquad
  \bar{\dot{u}} = \dot{u} + \frac{\partial f}{\partial q^k}\, \Q^k + \frac{\partial f}{\partial t} := \dot{u} + \dot{f}
\end{equation}
\end{subequations}

The manifold $\@j_1\/(P,\? \R)\@$ is naturally embedded into the
tangent space $\@T(P)\@$ through the identification
\[
j_1\/(P,\? \R)\ =\ \left\{\@ z \in T(P)\ \vline\ \left<\?z\?,\? dt\?\right>\, =\, 1\@ \right\}
\]
expressed in local coordinate as
\begin{equation}\label{Eq_1.17}
z \in j_1(P, \R)\quad \Longleftrightarrow\quad z\ =\ \left[\@\frac{\partial}{\partial t} + \Q^i(z) \frac{\partial}{\partial q^i} +
\U^i(z) \frac{\partial}{\partial u}\@\right]_{\pi(z)}
\end{equation}

In addition to the jet attributes, the space $\@j_1\/(P,\? \R)\@$
inherits from $\?P\?$ two distinguished actions of the group
$\@(\?\R,\? +\?)\@$, related in a straightforward way to the
identification \eqref{Eq_1.17}.

\vspace{2pt}
The former is simply the push--forward of the action
\eqref{Eq_1.14}, restricted to the submanifold $\@j_1\/(P,\? \R)
\subset T(P)\@$. In jet--coordinates, a comparison with equation
\eqref{Eq_1.17} provides the local representation
\begin{subequations}
\begin{equation}\label{Eq_1.18a}
  \left(\psi_{\xi}\right)_*\/(z)\ =\ \left[\@\frac{\partial}{\partial t}\, +\, \Q^i(z)\@ \frac{\partial}{\partial q^i}\, +\,
\U^i(z)\@ \frac{\partial}{\partial u}\@\right]_{\pi(z) + \xi}
\end{equation}
expressed symbolically as
\begin{equation}\label{Eq_1.18b}
  \left(\psi_{\xi}\right)_*\ :\ (\?t, q^i, u, \Q^i, \U^i\?)\ \longrightarrow\ (\?t, q^i, u + \xi, \Q^i, \U^i\?)
\end{equation}
\end{subequations}
The quotient of $\@j_1(P,\? \R)\@$ by this action is a $(2n +
2)$--dimensional manifold, denoted by $\@\L(\V)\@$. 
The quotient map makes $\@j_1\/(P,\? \R)\@$
into a principal fibre bundle over $\@\L(\V)\@$, with structural
group $\@(\?\R,\? +\?)\@$. Furthermore, by equation \eqref{Eq_1.18b}, 
$\@\L(\V)\@$ is an affine fibre bundle over $\@\V\@$ with
local coordinates $\@t, q^i, \Q^i, \U\@$.

\vspace{2pt}
The latter action of $\@(\?\R ,\? +\?)\@$ on $\@j_1\/(P,\? \R)\@$
follows from the invariant character of the field
$\@\tfrac{\partial}{\partial u}\@$ and is expressed in local
coordinates by the addition
\begin{subequations}
  \begin{equation}\label{Eq_1.19a}
    \phi_\xi\/(z)\ :=\ z + \xi\, \left(\frac{\partial}{\partial u}\right)_{\!\pi(z)} =\
    \left[\@\frac{\partial}{\partial t} + \Q^i(z)\@ \frac{\partial}{\partial q^i} +
    \left(\? \U^i(z) + \xi\? \right)\@ \frac{\partial}{\partial u}\@\right]_{\pi(z)}
  \end{equation}
  summarized into the symbolic relation
  \begin{equation}\label{Eq_1.19b}
  \phi_{\xi}\ :\ (\?t, q^i, u, \Q^i, \U^i\?)\ \longrightarrow\ (\?t, q^i, u, \Q^i, \U^i + \xi\?)
\end{equation}
\end{subequations}
The quotient of $\@j_1\/(P,\?\R)\@$ by this action is once again a
$(2n + 2)$--dimensional manifold, denoted by
$\?\L^c(\V)\@$. As before, equation \eqref{Eq_1.19b} points out the nature
of $\? \L^c(\V)\@$ as a fibre bundle over $\?P\?$ (\@as well as on
$\@\V\@$), with coordinates $\@t, q^i, u, \Q^i\@$. The quotient map
makes $\@j_1\/(P,\? \R) \to \L^c(\V)\@$ into a principal fibre
bundle, with structural group $\@(\?\R,\? +\?)\@$ and group action
\eqref{Eq_1.19a}.

\vspace{2pt}
Eventually, the
group actions \eqref{Eq_1.18a}, \eqref{Eq_1.19a} do \emph{commute}.
Therefore, each of them may be used to induce a group action on the
quotient space generated by the other one. 
This makes both $\?\L(\V)\@$ and $\?\L^c(\V)\@$ into
principal fibre bundles over a common ``double quotient'' space,
canonically diffeomorphic to the
velocity space $\?j_1(\V)\@$.

\smallskip
The situation is summarized into the commutative diagram

\begin{equation*}
\begin{CD}
  j_1(P, \R) @>>>  \L^c(\V) \\ @VVV @VVV \\ \L(\V) @>>> j_1(\V)
\end{CD}
\end{equation*}\\[2pt]
\noindent in which all arrows denote principal fibrations, with
structural groups isomorphic to $\@(\?\R,\? +\?)\@$ and group
actions obtained in a straightforward way from equations
\eqref{Eq_1.18b}, \eqref{Eq_1.19b}. The principal fibre bundles
$\L(\V) \to j_1(\V)$ and $\L^c(\V) \to j_1(\V)$ are respectively
called the \emph{Lagrangian} and the \emph{co--Lagrangian bundle}
over $\?j_1(\V)\@$.

\vspace{5pt}
The advantage of this framework is most appreciated by giving up the
traditional approach, based on the interpretation of the Lagrangian
function $\Lagr(t, q^i, \Q^i)$ as the representation of a\linebreak
(\@gauge--dependent\@) \emph{scalar field} over $\?j_1(\V)$ and
introducing instead the concept of \emph{Lagrangian section}, meant
as a section 
\[\ell\ \colon\ j_1(\V) \to \L(\V)\] 
of the Lagrangian bundle.

For each choice of the trivialization $\@u\@$ of $\?P\?$, the
description of $\?\ell\?$ takes the local form
\begin{equation}\label{Eq_1.20}
  \U\, =\, \Lagr(t, q^i, \Q^i)
\end{equation}
and so it does still rely on the assignment of a function
$\@\Lagr(t, q^i, \Q^i)\@$ over $\?j_1(\V)$. However, as soon as the
trivialization is changed into $\@\bar{u} = u + f\@$, the
representation \eqref{Eq_1.20} undergoes the transformation law
\begin{equation}\label{Eq_1.21}
\bar{\U}\, =\, \U + \dot{f}\, =\, \Lagr(t, q^i, \Q^i) + \dot{f} := \Lagr'(t,
q^i, \Q^i)
\end{equation}
involving a different, gauge--equivalent, Lagrangian.


\subsubsection{The non-holonomic Lagrangian bundles}
In the presence of non--holonomic constraints, the construction of the
Lagrangian bundles may be easily adapted to the submanifold $\@\A\@$,
through a straightforward pull--back process.

\vspace{2pt}
\noindent The situation is conveniently illustrated by means of a commutative
diagram
\\[1em]
\begin{equation}\label{Diag_1.22}
\begin{picture}(80,20)
\put (7,-25)    {\vector(1,0){42}}
\put (32,-10)   {\line(1,0){20}}
\put (54,-10)   {\vector(1,0){15}}
\put (9,5)      {\vector(1,0){36}}
\put (33,20)    {\vector(1,0){36}}
%
\put (0,2)      {\vector(0,-1){24}}
\put (53,2)     {\vector(0,-1){24}}
\put (23,17)    {\line(0,-1){11}}
\put (23,4)     {\vector(0,-1){11}}
\put (78,17)    {\vector(0,-1){24}}
%
\put (3,-22)    {\vector(2,1){18}}
\put (56,-22)   {\vector(2,1){18}}
\put (3,8)      {\vector(2,1){18}}
\put (56,8)     {\vector(2,1){18}}
%
\put (0,-25)    {\makebox(0,0){$\L(\A)$}}
\put (53,-25)   {\makebox(0,0){$\A$}}
\put (53,5)     {\makebox(0,0){$\L^c(\A)$}}
\put (0,5)      {\makebox(0,0){$j_1^{\A}(P, \R)$}}
\put (78,-10)   {\makebox(0,0){$j_1(\V)$}}
\put (78,20)    {\makebox(0,0){$\L^c(\V)$}}
\put (23,20)    {\makebox(0,0){$j_1(P, \R)$}}
\put (23,-10)   {\makebox(0,0){$\L(\V)$}}
\end{picture}
\end{equation}
\\[7em]
where:
\begin{itemize}
  \item $\L(\A)\@$ and $\?\L^c(\A)\@$ are respectively the pull--back of $\?\L(\V)\@$
and $\?\L^c(\V)\@$ on the submanifold $\?\A \to
j_1(\V)\@$;\vspace{1pt}
\item the space $j_1^{\A}\/(P,\? \R)\@$ may be alternatively seen as the pull--back
of the jet--bundle\linebreak $j_1\/(P,\? \R) \to \L(\V)\@$ on the
submanifold $\L(\A) \to \L(\V)$ or as the pull--back
of\linebreak $j_1\/(P,\? \R) \to \L^c(\V)$ on $\?\L^c(\A) \to
\L^c(\V)\@$.
\end{itemize}

\smallskip
\noindent The geometrical properties of the above--defined pull--back bundles
are straightforwardly inherited from their respective holonomic
counterparts.
In particular:
\begin{itemize}
  \item Every choice of a trivialization $\@u\@$ of $\?P\?$ allows to lift
  any coordinate system of $\A\@$ to coordinates $\@t, q^i, z^A,
  u\@$ on $\L^c(\A)\@$, $\@t, q^i, z^A, \U\@$ on $\L(\A)\@$ and
  $\@t, q^i, u, z^A, \U\@$ on $\?j_1^{\A}\/(P,\? \R)\@$. The
  resulting coordinate transformations are obtained by
  completing equations \eqref{Eq_1.5} with (\@the significant
  part of\@) the system
  \begin{equation}
    \bar{u} = u + f(t, q^1,\ldots, q^n)\ ,\quad \bar{\dot{u}} = \dot{u} + \frac{\partial f}{\partial t} + \frac{\partial f}{\partial q^k}\, \psi^k(t, q^i, z^A)
    := \dot{u} + \dot{f}
  \end{equation}
  \item The embeddings $\?\L(\A) \to
  \L(\V)\@$, $\@\L^c(\A) \to \L^c(\V)\,$ as well as
  $\@j_1^{\A}\/(P,\? \R) \to j_1(P, \R)$ are all locally
  described by equation \eqref{Eq_1.6a}\?.
  \vspace{1pt}
  \item Both actions \eqref{Eq_1.18a}, \eqref{Eq_1.19a} of the group $\@(\?\R,\? +\?)\@$ on $j_1\/(P,\? \R)\@$
  preserve the submanifold $\?j_1^\A\/(P,\? R)\@$ thereby
  inducing two corresponding actions $\?\left(\psi_\xi\right)_*
  \@$ and $\?\phi_\xi\@$ on $j_1^\A\/(P,\? R)\@$, expressed in
  coordinate as\vspace{1pt}
  \begin{subequations}
  \begin{align}
\left(\psi_\xi\right)_*\ &\colon\ \left(\?t, q^i, u, z^A,
\U\?\right)\quad \longrightarrow \quad \left(\?t, q^i, u + \xi, z^A,
\U\?\right)\label{Eq_1.24a}\\[4pt] \phi_\xi\ &\colon\ \left(\?t, q^i, u, z^A, \U\?\right)\quad
\longrightarrow \quad \left(\?t, q^i, u, z^A, \U + \xi\?\right)\label{Eq_1.24b}
  \end{align}
  \end{subequations}\vspace{1pt}
Acting in the same way as before, it is easily seen that the
manifold $\?j_1^\A\/(P,\? \R)\@$ is a principal fibre bundle
over $\L(\A)$ under the action $\left(\psi_\xi\right)_*\@$, as
well as a principal fibre bundle over $\L^c(\A)$ under the
action $\?\phi_\xi\@$. Moreover, both $\L(\A)$ and $\L^c(\A)$
are principal fibre bundles over $\A$ under the (\@induced\@)
actions $\?\left(\psi_\xi\right)_*\?$ and $\@\phi_\xi\@$
respectively. Accordingly, all arrows in the front and rear
faces of diagram \eqref{Diag_1.22} express principal fibrations,
while those in the left and right--hand faces are principal
bundle homomorphisms.
\end{itemize}

\noindent Preserving the terminology, the principal  bundles $\L(\A)
\to \A$ and $\L^c(\A) \to \A$ are respectively called
the \emph{non--holonomic Lagrangian bundle} and the \emph{non--holonomic
co--Lagrangian bundle} over $\A\@$, while any section $\ell \colon \A \to
\L(\A)$ is referred to as a \emph{non--holonomic Lagrangian section}.

\vspace{1pt}
Once a trivialization $\@u\@$ of $\?P\?$ has been fixed, the description of
 $\@\ell\@$ takes the local form
\begin{equation}\label{Eq_1.25}
  \U\ =\ \Lagr(t,q^i, z^A)
\end{equation}
which undergoes the transformation law
\begin{equation}\label{Eq_1.26}
  \bar{\U}\, =\, \U + \dot{f}
  \, =\, \Lagr(t, q^i, z^A) + \frac{\partial f}{\partial t} +
  \frac{\partial f}{\partial q^i}\, \psi^i\, :=\, \Lagr'(t, q^i, z^A)
\end{equation}
under an arbitrary change $\@u \to u +f\/(t, q^1, \ldots, q^n)\@$. 


\subsubsection{The Hamiltonian bundles}
Parallelling the discussion in \S\@\ref{SSSec_The Lagrangian
bundles}, we shall now briefly go over the construction of the
\emph{Hamiltonian bundles} on $\V\@$. To this end, we focus on the
fibration $\@P\to\V\/$, and denote by $\@\pi \colon j_1(P, \V) \to
P\,$ the associated first jet--space.

\noindent Every fibred coordinate system $\@t,q^i,u\/$ on $\?P\,$
induces local coordinates $\@t, q^i, u, p_0, p_i\@$ on $\?j_1\/(P,\?
\V)\@$, with transformation group
\begin{subequations}\label{Eq_1.27}
\begin{alignat}{2}
&\bar{t}=t+c\, , \qquad \bar{q}\@^i = \bar{q}\@^i(t, q^1, \ldots,
q^n)\, ,&\qquad &\bar{u} = u + f(t, q^1, \ldots,
q^n)\\[3pt]
&\bar{p}_0 = p_0 + \frac{\partial f}{\partial t} + \left(p_k +
\frac{\partial f} {\partial q^k}\right)\ \frac{\partial
q^k}{\partial t}\;\,,& \qquad &\bar{p}_i = \left(p_k +
\frac{\partial f}{\partial q^k} \right)\ \frac{\partial
q^k}{\partial \bar{q}^i}
\end{alignat}
\end{subequations}

The manifold $\?j_1\/(P,\? \V)\?$ is naturally imbedded into the
cotangent space $\@T^*(P)$ through the identification
\[
j_1\/(P,\? \V)\ =\ \left\{\@ \eta \in T^*(P)\ \vline\ \left<\?\eta\? ,\? \frac{\partial}{\partial u}\?\right>\, =\, 1\@\right\}
\]
expressed in local coordinate as
\begin{equation}\label{Eq_1.28}
\eta \in j_1(P, \V)\quad \Longleftrightarrow\quad \eta\, =\, \left[\@du\@ -\@ p_0\/(\eta)\@ dt\@ -\@ p_i\/(\eta)\@ dq^i\@\right]_{\pi(\eta)}
\end{equation}

\noindent Furthermore, %
the jet--bundle structure endows
$\?j_1\/(P,\? \V)\?$ with a contact bundle, locally generated by the
\emph{Liouville\/ $1$--form}
\begin{equation}\label{Eq_1.29}
  \tilde{\Theta}\ =\ d u\, -\, p_0\,d t\, -\, p_i\, d q^i
\end{equation}
which is an intrinsically defined 
object, as ensured by equations (\ref{Eq_1.27}\hspace{1pt}a,b).

\medskip Exactly as in the Lagrangian case, one can easily establish
two distinguished actions of the group $\@(\?\R,\?+\?)\/$ on
$\?j_1\/(P, \V)\@$, expressed locally as
\begin{subequations}
\begin{align}\label{Eq_1.30a}
\left(\psi_\xi\right)_*\/(\eta)\, &:=\, {\left(\psi_{-\xi}\right)_*}^*(\eta)\
=\ \left[\@du - p_0\/(\eta)\@ dt - p_i\/(\eta)\@ dq^i \@ \right]_{\pi(\eta) + \xi}\\[0.4em]
\label{Eq_1.30b} \phi_\xi\/(\eta)\, &:=\, \eta - \xi\,
(dt)_{\pi(\eta)}\, = \left[\@du - \left(\?p_0\/(\eta) +
\xi\?\right)\@ dt - p_i\/(\eta)\@ dq^i \@ \right]_{\pi(\eta)}
\end{align}
\end{subequations}

\medskip
\noindent In this connection, we point out that:
\begin{itemize}
\item The direct product of the actions (\eref{Eq_1.30a}, b) makes
$j_1\/(P,\? \V)\@$ into a principal fibre bundle over a
$(2n+1)$--dimensional base space $\Pi(\V)\@$, with coordinates
$\@t,q^i,p_i\@$, called the \emph{phase space}.\vspace{1pt}
\item In view of equations \eqref{Eq_1.2}, (\ref{Eq_1.27}\hspace{1pt}a,b), the phase space
$\Pi(\V)$ is readily seen as an affine bundle over $\V\@$, modelled on
$V^*(\V)\@$.\vspace{1pt}
\item The quotient of $j_1\/(P,\? \V)\@$ by the action \eqref{Eq_1.30a},
denoted by $\@\H(\V)\@$, is an affine bundle over $\V\@$,
modelled on the cotangent space $\@T^*(\V)\@$ and called the
\emph{Hamiltonian bundle}.\vspace{1pt}
\item Any trivialization $u \colon P \to \R$ allows to lift every local coordinate system
$\@t, q^1, \ldots, q^n\@$ on $\V$ to a corresponding one $\@t,
q^1, \ldots, q^n, p_0, p_1, \ldots, p_n\@$ on $\H(\V)\@$,
subject to the transformation law
    \begin{equation}
      \bar{p}_0 = p_0 + \frac{\partial f}{\partial t}\ ,\qquad \bar{p}_i = p_i + \frac{\partial f} {\partial q^i}
    \end{equation}
    further to a change of $\@u\@$ into $\@\bar{u} = u + f(t,
    q^1, \ldots, q^n)\@$.\vspace{1pt}
\item The quotient map makes $\@j_1\/(P,\? \V)\@$ into a principal fibre bundle
over $\H(\V)\@$, with structural group
$\@(\?\R,\?+\?)\@$ and fundamental vector
$\@\tfrac{\partial}{\partial u}\@$.\vspace{2pt}
\item The canonical Liouville
$1$--form \eqref{Eq_1.29} endows $j_1\/(P,\? \V) \to \H(\V)$
with a distinguished connection, called the \emph{Liouville
connection}.\vspace{2pt} 
\item The action \eqref{Eq_1.30b}
``passes to the quotient'', thereby making $\@\H(\V)\@$ into a
principal fibre bundle over the phase space
$\Pi(\V)$\@.\vspace{2pt}
\item The quotient of $j_1\/(P,\? \V)$ by
the action \eqref{Eq_1.30b}, denoted by $\@\H^c(\V)\@$, is a\linebreak
$(2n+2)$--dimensional manifold, with coordinates $\@t, q^i, u,
p_i\@$, called thek \emph{co--Hamiltonian bundle}. The
quotient map makes $j_1\/(P,\? \V)$ into a principal fibre
bundle over $\@\H^c(\V)\@$.\vspace{2pt} 
\item The action
\eqref{Eq_1.30a}, suitably transferred to $\H^c(\V)\@$, makes
the latter into a principal fibre bundle over $\Pi(\V)$\,.
\end{itemize}

\vspace{4pt}
\noindent The previous discussion is summarized into the commutative diagram
\begin{equation}\label{Diag_1.32}
\begin{CD}
   j_1\/(P,\? \V)        @>>>   \H^c(\V)  \\
   @VVV                 @VVV \\
   \H(\V)           @>>>    \Pi(\V)
\end{CD}
\end{equation}
all arrows denoting principal fibrations with structural group
isomorphic to $\R\@$.

\vspace{1pt} \noindent As implicit in the notation, the manifold 
$j_1\/(P,\@ \V)$ is indeed
identical to the pull--back of $\@\H^c(\V)\@$ over $\@\H(\V)\@$, as
well as the pull--back of $\@\H(\V)\@$ over $\@\H^c(\V)\@$.


\subsection{Further developments}\label{SSec_Further developments}
\textbf{(\/i\/)} \,The identifications \eqref{Eq_1.17}, \eqref{Eq_1.28} provide a
natural pairing between the fibres of the first jet--spaces
$j_1\/(P,\? \R) \rarw{\pi} P\@$ and $j_1\/(P,\? \V) \rarw{\pi} P\@$,
locally expressed as
\begin{equation}\label{Eq_1.33}
\left<\?z\?,\?\eta\?\right>\, =\, \left<\?\left[\@\frac{\partial}{\partial t} + \Q^i(z)\/\frac{\partial}{\partial q^i} +
\U(z)\/\frac{\partial}{\partial u}\@\right]_{\pi(z)}\! ,\? \left[\@du - p_0\/(\eta)\@dt - p_i\/(\eta)\@dq^i \@ \right]_{\pi(\eta)}\?
\right>\raisetag{-2pt}
\end{equation}
for all $z \in j_1\/(P,\? \R)\@,\, \eta \in j_1\/(P,\? \V)$
satisfying $\pi(z) = \pi(\eta)\@$.

In view of equations \eqref{Eq_1.18a}, \eqref{Eq_1.30a}, the
correspondence \eqref{Eq_1.33} satisfies the invariance property
\begin{equation}\label{Eq_1.34}
 \Big<\? (\psi_\xi)_*(z)\, ,\, (\psi_\xi)_*(\eta)\?\Big>\
  =\ \big<\? z\? ,\? \eta\? \big>
\end{equation}
thereby inducing an analogous pairing operation between the fibres
of the bundles $\L(\V) \to \V\@$ and $\@\H(\V) \to \V\@$, or ---
just the same
--- giving rise to a bi--affine map of the fibred product $\?\L(\V) \times_{\V}
\H(\V)\@$ onto $\R\@$, expressed in coordinates as
\begin{equation}\label{Eq_1.35}
  \zeta, \mu\quad \longrightarrow \quad F(\zeta, \mu)\, :=\, \U\/(\zeta)\@ -\@ p_0\/(\mu)\@ -\@p_i\/(\mu)\, \Q^i\/(\zeta)
\end{equation}

\medskip
\noindent \textbf{(\/ii\/)} \,Let $\@\mathcal{S}\@$ denote the submanifold of $\@\L(\V)
\times_{\V} \H(\V)\@$ described by the equation
\begin{equation}\label{Eq_1.36}
  \mathcal{S}\, =\, \left\{\@ (\?\zeta\?,\?\mu\?) \in \L(\V) \times_{\V} \H(\V)\ \
  \vline\ \ F\/(\zeta, \mu)\, =\, 0\@ \right\}
\end{equation}

\noindent A straightforward argument, based on equation
\eqref{Eq_1.35}, shows that the submanifold $\@\mathcal{S}\@$ is at
the same time a fibre bundle over $\L(\V)\@$ as well as over
$\@\H(\V)\@$.
The former case is made explicit by referring
$\@\mathcal{S}\@$ to local coordinates $\@t, q^i, \Q^{\?i}, \U,
p_i\@$, the $\@p_i$'\?s been regarded as fibre coordinates. The
latter circumstance is instead accounted for by referring
$\@\mathcal{S}\@$ to coordinates $\@t, q^i, \Q^{\?i}, p_0, p_i\@$,
related to the previous ones by the transformation
\[
\U\ =\ p_0\, +\, p_i\, \Q^{\?i}
\]
and with the $\@\Q^{\?i}$'\?s now playing the role of fibre coordinates.

\smallskip The restriction to the submanifold $\?\SS\?$ of the
action
\[
\phi_\xi\/(\zeta\@, \mu)\, :=\, \big(\?\phi_\xi(\zeta)\@, \phi_\xi(\mu)\?\big) \qquad \quad
\forall\, (\?\zeta\?,\?\mu\?) \in \L(\V) \times_{\V} \H(\V)
\]
makes the latter into a principal fibre bundle over a $(3\?n +
1)$--dimensional base space $\@\CC\/(j_1(\V))\?$, with coordinates
$\@t,q^i,\Q^i,p_i\@$.

Depending on the choice made for the local coordinates over
$\@\mathcal{S}\@$, the resulting group action may be expressed
symbolically either as
\begin{subequations}
\begin{equation}
  \phi_{\xi}\ :\ (\?t, q^i, \Q^i, \U, p_i\?)\ \longrightarrow\ (\?t, q^i, \Q^i, \U + \xi, p_i\?)
\end{equation}
or
\begin{equation}
  \phi_{\xi}\ :\ (\?t, q^i, \Q^i, p_0, p_i\?)\ \longrightarrow\ (\?t, q^i, \Q^i, p_0 + \xi, p_i\?)
\end{equation}
\end{subequations}

\vspace{0.3cm} \noindent The situation is summarized into the
following diagram\vspace{0.3cm}
\begin{equation}\label{Diag_1.38}
\xymatrix{ & \mathcal{S} \ar[rr] \ar[dd] \ar[dl] & &
\H\/(\V) \ar[dd] 
\\
\L\/(\V) \ar[dd] 
& & & 
\\
& \mathscr{C}\/(j_1\/(\V)) \ar[dl] \ar[rr]  & &
\Pi\/(\V) \ar[dl]
\\
 j_1\/(\V) \ar[rr] & &
\V &
 }
\end{equation}

\vspace{0.3cm}\noindent In view of equations \eqref{Eq_1.2} and
(\ref{Eq_1.27}\@b), the manifold $\?\CC\/(j_1(\V))\?$ is --- by
construction --- an affine bundle over $\@j_1(\V)\?$, modelled on the
contact space $\?\C\/(j_1(\V))\@$.

\medskip \noindent \textbf{(\/iii\/)} \, The canonical contact
$1$--form \eqref{Eq_1.29} of $\?j_1\/(P,\? \V)\@$ can be
pulled--back onto the fibred product $j_1(P, \R) \times_P j_1(P,
\V)\@$, thus endowing the principal fibre bundle $j_1(P, \R) \times_P j_1(P, \V)
\to \L(\V) \times_{\V} \H(\V)$  with a canonical connection.

For every choice of the trivialization $\@u\@$ of $P\to \V\@$, the
difference $\@du - \tilde{\Theta}\@$ is (\@the pull--back of\@)\linebreak a
$1$--form $\@\tilde{\Theta}_u\@$ on $\L(\V) \times_{\V} \H(\V)\@$,
locally expressed as
\begin{equation}\label{Eq_1.42}
  \tilde{\Theta}_u\, =\, p_0\, dt + p_i\, dq^i
\end{equation}
and subject to the transformation law
\begin{equation}\label{Eq_1.43}
  \tilde{\Theta}_{\bar{u}} = \left(p_0 + \frac{\partial f}{\partial t}\right)\, dt + \left(p_i + \frac{\partial f}{\partial q^i}\right)\, dq^i =
  \tilde{\Theta}_u + df
\end{equation}
under an arbitrary change $\@u \to \bar{u} = u + f(t, q^1, \ldots, q^n)\@$.

\vspace{2pt} The form $\@\tilde{\Theta}_u\@$ can now be once again
pulled--back onto $\@\mathcal{S}\@$. In this last step, depending on
the choice of the coordinates over $\@\mathcal{S}\@$, the resulting
$1$--form is locally expressed as
\begin{equation}\label{Eq_1.44}
  \Theta_u\, =\, p_0\, dt + p_i\, dq^i\, \equiv\, \U\, dt + p_i\, \left(\?dq^i - \Q^i\, dt \?\right)
\end{equation}
Hence, the submanifold $\?\mathcal{S}\?$ is provided with a
distinguished $1$--form $\@\Theta_u\@$, defined up to the
choice of the trivialization of $\?P\@$.

\medskip \noindent \textbf{(\/iv\/)} \, In the presence of non--holonomic constraints,
the left--hand face of 
\eqref{Diag_1.38}\@,
\begin{equation}\label{Diag_1.39}
  \begin{CD}
    \mathcal{S} @>{\pi_{\mathcal{S}}}>>\L(\V)\\
    @VVV @VVV\\
    \CC(j_1(\V)) @>>> j_1(\V)
  \end{CD}
\end{equation}
may be easily pulled--back through the imbedding
$\@\A\rarw{i}{j_1\/(\V)}\@$, giving rise to the analogous
diagram\vspace{2pt}
  \begin{equation}\label{Diag_1.40}
  \begin{CD}
    \mathcal{S}^\A @>{\pi_{\mathcal{S}}}>> \L(\A) \\
    @VVV @VVV\\
    \CC(\A) @>>> \A
  \end{CD}
\end{equation}\vspace{2pt}

\noindent By construction, the manifold $\@\mathcal{S}^\A\@$ is then
a principal fibre bundle over the base space $\?\CC(\A)\@$ under the
(\@induced\@) action
\begin{equation}\label{Eq_1.41}
  \phi_{\xi}\ :\ (t, q^i, z^A, \U, p_i)\ \longrightarrow\ (t, q^i, z^A, \U + \xi, p_i)
\end{equation}
while, in the same manner as before, $\@\CC\/(\A)$ is an affine
bundle over $\@\A\?$ modelled on the non--holonomic contact bundle
$\@\C\/(\A)\?$.

By means of the pull--back procedure, the canonical form
\eqref{Eq_1.44} determines a distinguished\linebreak $1$--form on
$\?\mathcal{S}^\A$, locally expressed by\@\footnote{The same symbol
$\Theta_u$ will stand for both the form \eqref{Eq_1.44} and its pull--back on $\A\@$.}
\begin{equation}\label{Eq_2.4}
  \Theta_u\, =\, p_0\, dt + p_i\, dq^i\, \equiv\, \U\, dt + p_i\, \left(dq^i - \psi^i\, dt\right)
\end{equation}

\medskip \noindent \textbf{(\/v\/)} \, Every non--holonomic Lagrangian 
section $\?\ell \colon \A \to
\L(\A)\@$ determines a trivialization\linebreak $\varphi_\ell \colon \L(\A)
\to \R\?$ of the bundle $\?\L(\A) \to \A\@$. Let $\hat{\varphi}_\ell
:= \pi_{\mathcal{S}}^*(\varphi_\ell)$ denote the pull--back of
$\?\varphi_\ell\?$ to $\@\mathcal{S}^\A\@$, locally expressed as
\begin{equation}
  \hat{\varphi}_\ell\/ (t, q^i, z^A, \U, p_i)\, =\, \varphi_\ell(t, q^i, z^A, \U)\, =\, \U\@ -\@ \Lagr(t, q^i, z^A)
\end{equation}
From this, taking equation \eqref{Eq_1.41} into account, it is an
easy matter to check that the function $\@\hat{\varphi}_\ell\@$ is a
trivialization of the bundle $\?\mathcal{S}^\A \to \CC(\A)\?$ and
that, as such, it determines a section $\@\tilde{\ell} \colon
\CC(\A) \to \mathcal{S}^\A\@$, locally described by the equation
\begin{equation}\label{Eq_1.49}
  \U\, =\, \Lagr(t, q^i, z^A)
\end{equation}

In brief, every section $\?\ell \colon \A \to \L(\A)\?$ may be
lifted to a section $\tilde{\ell} \colon \CC(\A) \to
\mathcal{S}^\A\@$. The local representations of both sections are
formally identical and they obey the transformation law
\eqref{Eq_1.26} for an arbitrary change of the trivialization $\@u
\colon P \to \R\@$.

\vspace{2pt}
The section $\@\tilde{\ell} \colon \CC(\A) \to \mathcal{S}^\A\@$ may
now be used to pull--back the form \eqref{Eq_2.4} onto $\CC(\A)\@$,
hereby getting the $1$--form
\begin{equation}\label{Eq_2.7}
  \ThPPC\, :=\, {\tilde{\ell}}^* (\Theta_u)\, =\, \Lagr\, dt \@+\@ p_i\, \left(\?dq^i - \psi^i\, dt\?\right)\, :=\,
  - \Ham\, dt \@+\@ p_i\, dq^i
\end{equation}
henceforth referred to as the \emph{Pontryagin--Poincar\'e--Cartan form}.

\vspace{1pt}
Needless to say, the difference $\@\Ham := p_i\, \psi^i - \Lagr\@$,
known in the literature as the \emph{Pontryagin Hamiltonian}, is not
an Hamiltonian in the traditional sense.

\begin{remark}
The nature of the Pontryagin Hamiltonian may be understood by
pointing up that, in view of equations \eqref{Eq_1.6}, \eqref{Eq_1.35}, \eqref{Eq_1.36},  
the space $\@\mathcal{S}^\A\@$ is --- by construction --- 
a submanifold of $\@\L(\V)
\times_{\V} \H(\V)\@$ locally described by the equations
\begin{equation*}
F\/(\zeta, \mu)\, =\,g^\sigma\/(t, q\/(\zeta), \Q\/(\zeta))\, =\,  0\quad ,\quad 
  \sigma = 1, \ldots, n - r\@ 
\end{equation*}
for any $\@(\?\zeta\?,\?\mu\?) \in \L(\V) \times_{\V} \H(\V)\@$.

\smallskip
Hence, the manifold $\@\mathcal{S}^\A\@$ may be equivalently referred to both local
coordinates $\@t,q^i,p_i,z^A,\U\@$ and $\@t,q^i,p_i,z^A,p_0\@$,
related one another by the transformation
\[
\U\ =\ p_0\, +\, p_i\, \psi^{\?i}\/(t, q^1, \ldots, q^n, z^1, \ldots, z^r)
\] 

In the former circumstance, the section $\@\tilde{\ell} \colon
\CC(\A) \to \mathcal{S}^\A\@$ is locally represented by 
equation \eqref{Eq_1.49}, while in the latter case its local representation involves
the Pontryagin Hamiltonian $\@\Ham\/(t, q^i, p_i, z^A)\@$ in the form
\begin{equation}
p_0\, =\, -\,\Ham\/(t, q^i, p_i, z^A)
\end{equation}
\end{remark}


\vspace{2pt}
\section{Application to Constrained Variational Calculus}
\subsection{Problem statement}\label{SSec_Problem statement}
As already mentioned, we shall consider a constrained abstract system $\?\mathcal{B}\?$.
Given a differentiable function $\@\Lagr\in\F\/(\A)\,$ on the
space $\A\@$ (\@called the \emph{Lagrangian\/}\@) and
denoted by $\@\GAMMA\@$ the lift to $\A\@$ of an admissible evolution 
$\@\gamma\@$ of the system, define
the \emph{action functional\/}
\begin{equation}\label{Eq_2.1}
\I\@[\gamma]:= \int_{\GAMMA} \Lagr\/(\/t, q^1, \ldots,q^n,z^1,\ldots,z^r\/)\,d\/t
\end{equation}

We intend to
characterize the extremals (\@if at all\@) of the functional
$\@\I\@[\gamma]\@$ among all the admissible evolutions $\@\gamma
\colon [\?t_0\?,\?t_1\?] \to \V\?$ connecting two given
configurations.

\vspace{2pt}
As we shall see, it turns out to be more convenient to approach the topic
by defining an alternative variational problem on the
manifold $\@\CC\/(\A)\?$ and proving that, under suitable hypotheses, 
the latter is equivalent to the original one.

\vspace{4pt}
\begin{remark}
The constrained variational problem based on the functional
(\ref{Eq_2.1}) may be also viewed as a  typical optimal control
 problem.

\vspace{.5pt} As in \S\@\ref{SSEc_Non-holonomic constraints}\@, the
admissibility of a given evolution of $\@\B\?$\vspace{.7pt} is
expressed in local coordinates by equation (\ref{Eq_1.7})\@,
expressing the derivatives $\tfrac{d q^i}{d t}\@$ in terms of a
smaller number of  variables $z^A,\;A=1,\ldots,r\@$.

Therefore, every concurrent assignment of both the values of
$z^1\/(t), \ldots, z^r\/(t)\@$ and of a point in the event space
$\?\V\@$ determines an admissible evolution of the system as the
solution of the ordinary differential equations (\ref{Eq_1.7}) with
the given initial conditions. This makes the $z^A$'s into the
controllers of the evolution 
and, as such, they
usually go in the literature under the name of \emph{controls}.

In this sense, the search for the curves $\@\GAMMA = \GAMMA\/(t)\@$
along which the functional (\ref{Eq_2.1}) takes its extremal values
may be equivalently seen as the one for those particular controls
which \emph{optimize} the evolution of the system.

Actually,in the absence of specific assumptions on the nature of the manifold
$\@\A\@$, the functions $\?z^A\/(t)\@$, in themselves, have no
invariant geometrical meaning. In this respect,
attention should rather be shifted on \emph{sections\/}
$\@\sigma:\V\to\A\,$, locally expressed as $z^A = z^A\/(t, q^1, \ldots, q^n\/)\@$.
\end{remark}


\subsubsection{The gauge structure}
Given any pair of $1$--forms
$\@\Lagr\@d\/t\@$\vspace{1pt} and $\Lagr\@'\@d\/t\@$ over $\@\A\@$,
their respective action integrals
$\@\I\@[\gamma]=\int_{\GAMMA}\Lagr\@d\/t\@$ and
$\@\I\@'\/[\gamma]=\int_{\GAMMA}\Lagr\@'\@d\/t\@$\vspace{.7pt} give
rise to the same extremal curves if the difference
$\?\big(\?\Lagr\@' - \Lagr\?)\,d\/t$ is an \emph{exact}\vspace{1pt}
differential. 

This is easily seen as, under this circumstance, the equality $\@\oint
\Lagr\@d\/t = \oint \Lagr\@'\@d\/t$ holds along any closed curve,
thereby entailing the relation
\[
\I\@'\/[\gamma_\xi]\, -\, \I\/[\gamma_\xi]\, =\, \int_{\GAMMA_\xi} \big(\?
\Lagr\@' - \Lagr\?\big)\@d\/t \equiv \int_{\GAMMA} \big(\?
\Lagr\@' - \Lagr\?\big)\@d\/t\
\]
for any deformation $\@\gamma_\xi\@$ vanishing at the end--points,
whence also
\[
\d /d \xi\@ \Big(\?\I\@'\/[\gamma_\xi]\, -\,
\I\/[\gamma_\xi]\?\Big)\, \equiv\, 0
\]

\vspace{4pt}
As a consequence, as far as a variational problem based on
the functional \eqref{Eq_2.1} is concerned, the Lagrangian function
$\@\Lagr \in \F(\A)\@$ is defined up to an equivalence relation of
the form
\begin{equation}\label{Eq_2.2}
  \Lagr\ \sim\ \Lagr\@'\qquad \Longleftrightarrow\qquad \Lagr\@'\, -\, \Lagr\ =\ \d f /d t\ ,\quad f \in \F(\V)
\end{equation}
Otherwise stated, the real information is not brought so much by
$\@\Lagr\@$ in itself as by a whole family of Lagrangians,
equivalent to each other in the sense expressed by equation
\eqref{Eq_2.2}.

The significance of the arguments developed so far relies actually
on the fact, explicitly pointed out by equations \eqref{Eq_1.25},
\eqref{Eq_1.26}, that the representation of an arbitrary section
$\@\ell \colon \A\to \L(\A)\@$ involves exactly this family of
Lagrangians, henceforth denoted by $\@\Lambda(\ell)\@$. A
straightforward check shows that a necessary and sufficient
condition for two sections $\@\ell\@$ and $\@\ell'\@$ to fulfil
$\@\Lambda(\ell) = \Lambda(\ell')\@$ is that the difference $\@\ell'
- \ell\@$, viewed as a function over $\@\A\@$, be itself of the form
\begin{equation}\label{Eq_2.3}
 \ell'\, -\, \ell\ =\ \d f /d t\ ,\qquad f \in \F(\V)
\end{equation}
Thus we see that, within our geometrical framework, the equivalence
relation \eqref{Eq_2.2} between \emph{functions} is replaced by the
almost identical relation \eqref{Eq_2.3} between \emph{sections}.
Intuitively, the latter is a sort of ``active counterpart'' of the
transformation law \eqref{Eq_1.26} for the \emph{representation} of
a given section $\@\ell\@$ under arbitrary changes of the
trivialization $u \colon P \to \R$.

\subsection{Extremals}\label{SSec_Extremals}
To start with, we observe that the algorithm described
in \S\@\ref{SSec_Further developments} allows to deduce 
--- in a \emph{canonical} way ---
 a Pontryagin--Poincar\'e--Cartan form
\begin{equation}\label{Eq_2.4vera}
  \ThPPC\, =\, \Lagr\, dt \@+\@ p_i\, \left(\?dq^i - \psi^i\, dt\?\right)\, =\,
  - \Ham\, dt \@+\@ p_i\, dq^i
\end{equation}
on the manifold $\@\CC(\A)\@$ from the only knowledge of the input data of the problem, namely
\begin{itemize}
  \item[\textit{i)}] the non--holonomic constraints,
 described by the imbedding $i \colon \A \to j_1\/(\V)$ and
 locally expressed by the equations \eqref{Eq_1.6a}\@;\vspace{2pt}
  \item[\textit{ii)}] the non--holonomic Lagrangian section  $\@\ell\@$, locally represented in
  the form \eqref{Eq_1.25}, being $\@\Lagr\/(t, q^i, z^A)\@$ the given Lagrangian function.
\end{itemize}

\medskip To understand the role of the form \eqref{Eq_2.4vera} in the present 
constrained variational context, we next focus
on the fibration $\CC(\A)\rarw{\u}\V$, given by the composite
map $\@\CC(\A)\to\Pi\/(\V)\to\V$ 
and we define an action integral over $\@\CC\/(\A)\@$, assigning to
each section $\@\tgamma:\@\V\to\CC(\A)$, locally expressed by
$q^i=q^i\/(t)\@,\,z^A=z^A\/(t)\@, p_i=p_i\/(t)\@$, the
real number
\begin{equation}\label{Eq_2.8}
\I\@[\@\tgamma\@]\,:=\,\int_{\tgamma}\?\ThPPC\,=\,\int_{t_0}^{t_1}\!
\left(\@p_i\,\frac{d q^i}{dt} -\Ham\@\right)d\/t
\end{equation}

\vspace{2pt} Given any deformation $\tgamma_\xi\@$\vspace{1pt}
preserving the end--points of $\?\u\cdot\tgamma\@$ and indicating
with\linebreak $\@\tilde{X}=X^i\/(t)\@\big(\de /de
{q^i}\big)_{\tgamma}\?+\? \@\Gamma^A\/(t)\,\big(\de /de
{z^A}\big)_{\tgamma}\?+ \@ \Pi_i\/(t)\@\big(\de /de
{p_i}\big)_{\tgamma}$ the corresponding infinitesimal deformation,
we get the relation
\begin{multline*}
\d{\@\I\@[\@\tilde{\gamma}_\xi\@]} /d
{\xi}\,\bigg|_{\@\xi=0}\,=\@\int_{t_0}^{t_1} \left[\bigg(\@\d\?{q^i}
/d t\@-\@\de \Ham /de {p_i}\bigg)\@\Pi_{\?i}\@- \@\bigg(\@\d {\?p_i}
/d t\@+\@\de {\Ham} /de {q^i}\bigg)\@X^i\@- \@\de {\Ham} /de
{z^A}\,\Gamma^A\@\right]d\/t
\end{multline*}

From the latter, taking the conditions $\?X^i\/(t_0)=X^i\/(t_1)=0\@$
into account, it is easy to conclude that the vanishing of $\@\d\@\I
/d {\xi}\@\big|_{\@\xi=0}\@$ under arbitrary deformations of the
given class is mathematically equivalent to the system\vspace{.5pt}
\begin{subequations}
\begin{align}
  &\d\/q^i/d t\,=\,\de\?\Ham/de{p_i}\,=\,\psi^i(t, q^i, z^A)\label{Eq_2.9a}\\[3pt]
  & \d\/p_i/d t\,=\,-\@\de\?\Ham/de{q^i}\,=\, -\@ p_k\@\de\? {\psi^k} /de {q^i}\@ +\@ \de\? {\Lagr}/de {q^i}\\[3pt]
  &\de\?\Ham/de{z^A}\,=\, p_i\@\de\? {\psi^i} /de {z^A}\@ -\@ \de\? {\Lagr}/de {z^A}\, =\, 0
\end{align}
\end{subequations}

\medskip \noindent Equation \eqref{Eq_2.9a} shows that the extremal curves of
the functional \eqref{Eq_2.8} are kinematically admissible. As they
are extremals with respect to arbitrary deformations vanishing at
the end--points, this automatically makes them extremals with
respect to the narrower class of admissible deformations as well.

Therefore, we can state that every ``\@free\/'' extremal of the
functional \eqref{Eq_2.8} projects onto an extremal curve $\gamma
\colon\/q^i = q^i(t)$ of the original problem.

\medskip
Conversely, given any admissible 
solution $\@\g\@$ of the assigned variational problem, it seems beforehand hard to establish if 
and under which hypotheses there exists at least one 
extremal $\@\tgamma\@$ of the functional \eqref{Eq_2.8} projecting onto $\g\@$.
Heuristically, the associated variational problem in $\?\CC\/(\A)\@$
may be viewed as the study of the functional \eqref{Eq_2.1} where
the kinematical admissibility condition \eqref{Eq_1.7} does not play
the role of an \emph{a priori} request upon sections anymore but is
instead retrieved afterwards by the method of Lagrange
multipliers. As a consequence, it can be reasonable that,
under suitable hypotheses, the equivalence between the
two variational problems could be proved. 
Let us investigate this point.

\smallskip
Referring to \cite{MBP}\ for details, we recall that an admissible
section $\gamma \colon [\@t_0\? ,\? t_1] \to \V$ is called
\emph{ordinary} if and only if all its infinitesimal deformations
vanishing at the end--points are tangent to a finite deformation
with fixed end--points. A remarkable subclass of the ordinary
sections is formed by the \emph{normal} ones. We will dedicate the
final paragraph to discuss their role in the present
context.\\

\smallskip \noindent We now state the following
\begin{theorem}\label{Th_2.1}
  Every ordinary extremal $\?\gamma\@$ of the functional (\ref{Eq_2.1}) is the projection of at least one extremal
  $\?\tgamma\@$ of the functional (\ref{Eq_2.8})\@.
\end{theorem}
\begin{proof}
 Given an ordinary extremal $\@\gamma\@$ of (\ref{Eq_2.1}), let $\GAMMA\, \colon q^i = q^i(t),\, z^A = z^A(t)$ be its lift into $\A\@$. As discussed in \S\ref{SSec_Infinitesimal deformations of sections}\@, an admissible infinitesimal deformation of $\GAMMA$ is then a vector field\linebreak $\@\hat{X}=X^i\/(t)\@\big(\de /de
{q^i}\big)_{\GAMMA}+ \@\Gamma^A\/(t)\,\big(\de /de
{z^A}\big)_{\GAMMA}\@$ satisfying the variational equation
\begin{equation}\label{Eq_2.10}
\d {X^i} /d t\,=\,\bigg(\de {\psi^i} /de
{q^k}\bigg)_{\!\GAMMA}\?X^k\, +\, \bigg(\de {\psi^i} /de
{z^A}\bigg)_{\!\GAMMA}\@\Gamma^A
\end{equation}
more suitably written as
\begin{equation}
\d\, /d t\@ \big({A^i}_j\, X^j\big)\, =\, {A^i}_j\, \de
{\psi^j} /de {z^A}\ \Gamma^A
\end{equation}
being $\@{A^i}_{\!j}(t)\@$ any non--singular solution of the matrix
equation
\begin{equation}\label{Eq_2.12}
  \d\/ {{A^i}_j} /d t\, +\, {A^i}_k\, \bigg(\de {\/\psi^k} /de {q^j}\bigg)_{\GAMMA} =\ 0
\end{equation}

\smallskip \noindent
Furthermore, setting $\@X^i(t_0) = 0\@$, the validity of the
relation
\begin{equation}\label{Eq_2.13}
  0\, =\, \int_{t_0}^{t_1} \!{A^i}_j\, \bigg(\de\/ {\psi^j} /de {z^A}\bigg)_{\GAMMA}\, \Gamma^A\ d\/t
\end{equation}
is then a necessary and sufficient condition for $\?X^i(t_1)\@$ to
vanish as well. The whole class of the infinitesimal deformations of
$\@\GAMMA\@$ vanishing at its end--points is therefore in
bijective correspondence with the totality of vertical vector fields
$\@\Gamma^A\/(t)\,\big(\de /de {z^A}\big)_{\GAMMA}\@$ defined along
$\?\GAMMA\@$ and satisfying the condition \eqref{Eq_2.13}\@.

\medskip
We next introduce
$\@n\?$ new functions $\?p_i = p_i(t)\@$
subject to the conditions
\begin{equation}\label{Eq_2.14}
  \d\/ {p_i} /d t\, +\, p_k\@ \bigg(\de\, {\psi^k} /de {q^i}\bigg)_{\GAMMA}\, =\,
  \bigg(\de {\Lagr} /de {q^i}\bigg)_{\GAMMA}
\end{equation}
In view of equation \eqref{Eq_2.12}\@, the functions $\@p_i\/(t)\@$ 
are defined up to a transformation
\begin{equation}
  p_i(t)\, \longrightarrow\, \bar{p}_i\/(t)\,=\, p_i\/(t)\, +\, \beta_j\@ {A^j}_i (t)
\end{equation}
with  $\big(\?\beta_1, \ldots, \beta_n\?\big) \in \R^n$. Therefore,
the thesis is proved as soon as we show that the stated hypotheses
entail the existence of at least
one choice of $\?\big(\?\beta_1, \ldots, \beta_n\?\big)\@$ such
that the resulting functions $\@\bar{p}_i\/(t)\@$
could be used to lift the curve $\GAMMA$ to an extremal
$\tgamma \colon \left[\?t_0\?,\?t_1\right] \to \CC(\A)\@$
of the functional \eqref{Eq_2.8}.

\medskip
Now, the extremality of the ordinary evolution $\?\gamma\@$
means that
\begin{equation}\label{Eq_2.16}
 \d{\@\I\@[\@\g_\xi\@]} /d
{\xi}\,\bigg|_{\@\xi=0}\,=\, \int_{t_0}^{t_1} \!\bigg( X^i\, \de \Lagr /de {q^i}\ +\ \Gamma^A\, \de \Lagr /de {z^A}\?\bigg)_{\GAMMA}\ d\/t\ =\ 0
\end{equation}
for every infinitesimal deformation of $\@\GAMMA\@$
vanishing at its end--points.

\smallskip \noindent
Moreover, taking the conditions $\@X^i(t_0) = X^i(t_1) = 0\@$ as
well as the equations \eqref{Eq_2.10} into account, we have
\begin{equation*}
  \begin{split}
    \int_{t_0}^{t_1}\!\!X^i\@ \bigg(\de \Lagr /de {q^i}\bigg)_{\GAMMA}\, d\/t\, &= 
    \int_{t_0}^{t_1}\!\!X^i\@ \bigg(\d\/ {p_i} /d t\, +\, p_k\@
    \de {\psi^k} /de {q^i}\bigg)_{\GAMMA}\, d\/t =\\[1em]
    &= \int_{t_0}^{t_1} \!\!p_k\@ \bigg(-\d\/ {X^k} /d t\, +\,
    X^i\@ \de {\psi^k} /de {q^i}\bigg)_{\GAMMA}\, d\/t\, =
     -\@ \int_{t_0}^{t_1} \!\!p_k\, \bigg(\de {\psi^k} /de {z^A}\bigg)_{\GAMMA}\, \Gamma^A\, d\/t
  \end{split}
\end{equation*}
Substituting into equation \eqref{Eq_2.16}, we conclude that
\begin{equation}
  \int_{t_0}^{t_1} \!\bigg[ -p_k\, \bigg(\de {\psi^k} /de {z^A}\bigg)_{\GAMMA}\, +\,
  \bigg(\de \Lagr /de {z^A}\bigg)_{\GAMMA}\ \bigg]\, \Gamma^A\ d\/t\ =\ 0
\end{equation}
for each $\@\Gamma^A = \Gamma^A(t)\@$ fulfilling the condition
\eqref{Eq_2.13}. In order for this to happen it is then necessary and
sufficient the validity of the linear relation
\[
-\@ p_k\, \bigg(\de {\psi^k} /de {z^A}\bigg)_{\GAMMA}\,
+\,
  \bigg(\de \Lagr /de {z^A}\bigg)_{\GAMMA}\, =\, \beta_k\@ {A^k}_j\@
  \bigg(\de {\psi^j} /de {z^A}\bigg)_{\GAMMA}
\]
or --- what is the same --- the existence of at least one solution
$\@ \bar{p}_i(t)\, =\, p_i(t)\, +\, \beta_j\@ {A^j}_i\@$ of the
system \eqref{Eq_2.14} such that
\begin{equation}\label{Eq_2.18}
  \bar{p}_k\, \bigg(\de {\psi^k} /de {z^A}\bigg)_{\GAMMA}\, -\,
  \bigg(\de \Lagr /de {z^A}\bigg)_{\GAMMA}\, =\, 0
\end{equation}
Equations \eqref{Eq_2.14} and \eqref{Eq_2.18}, together with the
kinematical admissibility condition (\@true by hypothesis\@) are
exactly the Euler--Lagrange equations for the functional
\eqref{Eq_2.8}. Therefore, the resulting curve $\?\tgamma \colon q^i
= q^i(t),\, z^A = z^A(t),\, \bar{p}_i = \bar{p}_i(t)\@$ is an extremal of the
functional \eqref{Eq_2.8} which projects onto $\?\gamma\@$.
\end{proof}

\smallskip
Collecting all results, we have just shown that, as far as the \emph{ordinary\/} 
evolutions are concerned, the
original constrained variational problem in the event space is
 equivalent to a canonically associated \emph{free\/} one in 
$\@\CC\/(\A)\@$.

\bigskip \noindent
\begin{remark}[Same problem, equivalent solution] 
There is another possible approach to the problem outlined in \S\@\ref{SSec_Problem statement}, which is slightly different but completely equivalent to the one taken so far.
Apparently it just complicates matters without giving any significant
advantage. On the other hand, it seems to be the most faithful
translation of the original Pontryagin's treatment of the topic
\cite{Pontryagin} into the present geometrical context. Hence --- at least for
historical reasons --- it is worth telling about.

\bigskip \noindent The algorithm relies on the following considerations:

\medskip \noindent \textbf{(\/i\/)} \,Given any section $\gamma \colon \left[\@t_0\?,\?
t_1\?\right] \to P$ of the bundle of affine scalars, let $\GAMMA$ denote the restriction to 
$\?\L\/(\V)\@$ of its jet--extension. 

\smallskip
The input data of the assigned problem
are taken into account by the introduction of a notion of \emph{admissibility} for
 $\@\g\@$. This is accomplished by requiring the jet--extension $\@\GAMMA\@$ 
to belong to a submanifold
$\@\AA\@$ of $\?\L(\V)\@$, locally described by the equations
\begin{equation}\label{Eq_2.20}
  \Q^i = \psi^i(t, q^i, z^A)\ ,\quad \U = \Lagr(t, q^i, z^A)
\end{equation}

In other words, the simultaneous assignment of
both the kinetic constraints and the Lagrangian function are used to express
the submanifold $\@\AA\@$ as the image
$\@\ell(\A)\@$. 

\smallskip
In this way, every admissible section $\@q^i =
q^i(t)\@$ in $\V$ determines, up to the initial value $\@u\/(t_0)\@$, a corresponding admissible
section $\@q^i =q^i(t),\, u = u(t)\@$ of $P\@$.

\smallskip
Compared with the main approach,
the present formulation just replaces the section $\?\ell \colon \A
\to \L(\A)\@$ with the image space $\AA = \ell(\A)\@$, viewed as
a submanifold of $\@\L(\V)\@$. 
As we have seen, the submanifold $\@\AA\@$ and, consequently,
the section $\?\ell\@$ are regarded
 as data of the problem and therefore the representation $\@\U = \Lagr(t, q^i,
z^A)\@$ is affected only by passive gauge transformations. The same
variational problem is therefore bred by different submanifolds related
one another by the action of the gauge group.

\medskip \noindent \textbf{(\/ii\/)} \,  
Let us consider the constrained variational problem with fixed end--points 
on the manifold $\L(\V)\@$\vspace{.5pt}, based on the
functional
\begin{equation}\label{Eq_2.19}
  \I\/[\gamma] := \int_{\GAMMA}\, \!\U\ d\/t
\end{equation}
$\GAMMA$ denoting the jet--extension of an admissible section $\gamma \colon \left[\@t_0\?,\?
t_1\?\right] \to P$. The stated problem does not
depend on a particular choice of the gauge, as the $1$--form $\U\, dt$ is well--defined in
$\L\/(\V)\@$\vspace{.3pt} up to a term $\@\dot{f}\, dt\@$.

\medskip \noindent
In local coordinates, setting $\@\gamma \colon\, q^i = q^i(t)\, ,\ u = u(t)\@$, we have
\[
\int_{\GAMMA}\, \!\U\ d\/t = u(t_1) - u(t_0)
\]
and so, being the values of  $q^i(t_0)$ and $q^i(t_1)$ already
fixed by the boundary conditions, the problem consists in finding a curve
$\gamma$ over which the increment $u(t_1) - u(t_0)$ becomes stationary
and whose projection onto $\@\V\@$ joins the assigned end--points.

\medskip \noindent \textbf{(\/iii\/)} \,The submanifold $\@\AA
\to \L(\V)\@$ is lifted up onto a submanifold $\@\CC(\AA) \to
\mathcal{S}\@$ whether by identifying $\?\CC(\AA)\@$ with the image
space $\@\tilde{\ell}(\CC(\A))\@$ or by pulling back
$\@\mathcal{S}\@$ onto $\@\AA\@$ by means of the commutative diagram
\begin{equation}
  \begin{CD}
    {\CC(\AA)}   @>{\tilde{\jmath}}>>  {\mathcal{S}} \\
    @VVV                               @VVV \\
    {\AA}       @>{\jmath}>>         {\L(\V)}
  \end{CD}
\end{equation}

All the same, the embedding $\@\CC(\AA) \rarw{\tilde{\jmath}}
\mathcal{S}\@$ is fibred onto $\@\Pi(\V)\@$ and its expression in
coordinate is formally identical to equations \eqref{Eq_2.20} which
are involved in the representation of the
submanifold $\@\AA\@$.

\smallskip
It is now possible to make use of the form \eqref{Eq_1.44} to provide
the manifold $\?\CC(\AA)\@$ with the $1$--form
\begin{equation}\label{Eq_2.22}
  \tilde{\jmath}^{\@ *} \left(\Theta_u\right)\, =\, \Lagr\, dt + p_i\, \left(dq^i - \psi^i\, dt\right)
\end{equation}
and, as a consequence, to define an action integral
by the integration of \eqref{Eq_2.22} along any section
$\check{\gamma} \colon [\@t_0\?,\?t_1\?] \to \CC(\AA)$.

Once again, this merely reproduces in the image space $\?\CC(\AA)\@$
the construction previously carried on in \S~\!\ref{SSec_Further developments}. In other words,
the $1$--form \eqref{Eq_2.22} is simply the image of the
\emph{Pontryagin--Poincar\'e--Cartan} form \eqref{Eq_2.7} under the
diffeomorphism $\@\tilde{\ell} \colon \CC(\A) \to \CC(\AA)\@$.
\end{remark}


\subsection{The case of normal sections}
A significant role in the study of the variational problem based on
the functional \eqref{Eq_2.8} is played by the choice of the
Lagrangian section as $\?\U = \dot{f}(t,q)\@$\vspace{0.7pt}, which
is related to the particular
situation determined by the ansatz $\@\Lagr\/(t, q^i, z^A) = 0\@$.  

Under the stated circumstance, the functional
\begin{equation}\label{Eq_2.23}
\I_0\@[\@\tgamma\@]\,:=\,\int_{t_0}^{t_1}\!
\bigg(\@p_i\,\frac{d q^i}{dt} -\psi^i\@\bigg)d\/t
\end{equation}
may therefore be used to single out a purely geometrical variational
problem in the manifold $\@\CC\/(\A)\@$.

\smallskip \noindent The corresponding extremal curves are easily seen to
satisfy the Euler--Lagrange equations
\begin{equation}\label{Eq_2.24}
     \d\/{q^i} /d t\, =\, \psi^i\/(t, q^i, z^A)\ , \quad
     \d\/{p_i} /d t\,=\, -\@ p_k\@\de\? {\psi^k} /de {q^i}\ ,\quad
     p_k\@\de\? {\psi^k} /de {z^A}\, =\,0
\end{equation}

\medskip In view of equation \eqref{Eq_1.7}, every admissible evolution $\@\gamma \colon
[\@t_0\?,\?t_1\?]\to \V\@$ is the projection of at least
one extremal $\@\tgamma\@$ of the functional \eqref{Eq_2.23}. 
For this reason, the projection of $\@\tgamma\@$ under the map
$\@\CC\/(\A)\to\A\@$ coincides with the lift
$\@\hat{\gamma}:[\@t_0,t_1\@]\to\A\@$.

The extremals of \eqref{Eq_2.23} 
are therefore in a bijective correspondence with the solutions
$\@p_i\/(t)\@$ of the homogeneous system \eqref{Eq_2.24}, with the
functions $\@q^i\/(t),\@z^A\/(t)\@$ regarded as given. In other
words, the totality of extremals $\?\tgamma\@$ of the functional
\eqref{Eq_2.23} form a finite dimensional vector space over
$\?\R\@$, whose dimension will be referred to as the
\emph{abnormality index} of $\@\gamma := \u\?\cdot\? \tgamma\@$.

\smallskip
\begin{definition}
  An admissible curve  $\@\gamma \colon [\@t_0\?,\?t_1\?]\to \V\@$ is called
  \emph{normal} whenever its abnormality index vanishes. Otherwise $\?\gamma\@$ is said \emph{abnormal}.
\end{definition}

As a consequence, a \emph{normal} curve $\@\gamma \colon
[\@t_0\?,\?t_1\?]\to \V\@$ is the projection of a \emph{unique}
extremal of the functional \eqref{Eq_2.23}, namely of the curve
\[
\tgamma\, \colon\ q^i = q^i(t)\ ,\quad z^A = z^A(t)\
,\quad p_i(t)\@ =\@ 0
\]

\bigskip Coming back to the study of the main variational problem
based on the action integral \eqref{Eq_2.1}, we can now state

\begin{theorem}
  Every normal extremal $\@\gamma\?$ of the functional (\ref{Eq_2.1}) is the projection of exactly one extremal
  $\@\tgamma\?$ of the functional (\ref{Eq_2.8}).
\end{theorem}
\begin{proof}
  The crucial point, that we will not demonstrate here, is that the
  normal curves form a subset of the ordinary ones.\footnote{See \cite{MBP}, Appendix B.}

  \noindent In view of Theorem \ref{Th_2.1}, this entails the extremal
  $\@\gamma\?$ to be the projection of at least one extremal $\@\tgamma \colon q^i =
  q^i(t),\, z^A = z^A(t),\,  p_i = p_i(t)\@$ of the functional
  \eqref{Eq_2.8}. We still have to prove its uniqueness.

  \smallskip \noindent To this end, we suppose the existence of a second extremal
  projecting onto $\@\gamma\?$, expressed in coordinate as
  ${\tgamma}\@' \colon q^i =
  q^i(t),\, z^A = z^A(t),\,  p_i = \tau_i(t)$.

  Taking equations \eqref{Eq_2.24}
  into account, it is easily seen that the contemporaneous
  validity of the  Euler--Lagrange equations for both the curves
  $\@{\tgamma}\@$ and $\@{\tgamma}\@'\@$
\begin{equation*}
     \d\/{q^i} /d t\, =\, \psi^i(t, q^i, z^A)\ ,\quad
     \d\/p_i /d t\@ +\@ p_k\@\de\? {\psi^k} /de {q^i}\, =\, \de\? {\Lagr} /de {q^i}\ ,\quad
     p_k\@\de\? {\psi^k} /de {z^A}\, =\, \de\? \Lagr /de {z^A}
\end{equation*}
\begin{equation*}
     \d\/{q^i} /d t\, =\, \psi^i(t, q^i, z^A)\ ,\quad
     \d\/\tau_i /d t\@ +\@ \tau_k\@\de\? {\psi^k} /de {q^i}\, =\, \de\? {\Lagr} /de {q^i}\ ,\quad
     \tau_k\@\de\? {\psi^k} /de {z^A}\, =\, \de\? \Lagr /de {z^A}
\end{equation*}
 makes the curve $\@q^i=q^i(t),\@z^A=z^A(t),\@p_i=p_i(t)\/-\/\tau_i(t)\@$ an
  extremal of the functional \eqref{Eq_2.23}.

  Since $\@\gamma\?$ is --- by hypothesis --- a \emph{normal} curve, this
  implies
  $\@p_i(t) = \tau_i(t)\@$ and hence $\@{\tgamma} \equiv
  {\tgamma}\@'$.
\end{proof}




\begin{thebibliography}{99}

\bibitem{MBP} E.~Massa, D.~Bruno and E.~Pagani, Geometric control theory I:
              mathematical foundations,
              {\it arXiv}\@:\@0705.2362v2 [math.OC].

\bibitem{MPL} E.~Massa, E.~Pagani and P.~Lorenzoni,
              On the gauge structure of Classical Mechanics,
              {\it Transport Theory and Statistical Physics} {\bf 29}, 69--91 (2000).

\bibitem{MVB} E.~Massa, S.~Vignolo and D.~Bruno, Non--holonomic Lagrangian and
              Hamiltonian Mechanics: an intrinsic approach,
              {\it J. Phys. A: Math. Gen.} {\bf 35}, 6713--6742 (2002).
              
\bibitem{Saunders} D.J.~Saunders, {\it The Geometry of Jet Bundles\/},
              London Mathematical Society, Lecture Note Series 142,
              Cambridge University Press (1989).                         

\bibitem{Bliss} G.A.~Bliss, {\it Lectures on the calculus of the variations\/},
              The University of Chicago Press, Chicago (1946).

\bibitem{Pontryagin} L.S.~Pontryagin, V.G.~Boltyanskii, R.V.~Gamkrelidze and E.F.~Mishchenko,
               {\it The mathematical theory of optimal process\/},
               Interscience, New York (1962).

\bibitem{Gelfand} I.M.~Gelfand  and S.V.~Fomin, {\it  Calculus of variations\/},
               Prentice-Hall Inc., Englewood Cliffs (1963).

\bibitem{Hestenes} M.R.~Hestenes, {\it Calculus of variations and optimal control theory\/},
               Wiley, New York London Sydney (1966).

\bibitem{Young} L.~C.~Young {\it Lectures on the Calculus of Variations and Optimal Control
               Theory\/} (second edition), AMS Chelsea Publishing, New York (1980).

\bibitem{Giaquinta} M.~Giaquinta and S.~Hildebrandt, {\it Calculus of variations I, II\/},
               Springer-Verlag, Berlin Heidelberg New York (1996).

\bibitem{Agrachev1} A.~A.~Agrachev and Yu.L.~Sachov, {\it Control Theory from the Geometric
               Viewpoint\/}, Springer-Verlag, Berlin Heidelberg New York (2004).

\bibitem{Arnold} V.~I.~ Arnold, {\it Dynamical Systems III, Encyclopaedia of
            Mathematical Sciences}, Springer-Verlag, Berlin Heidelberg New York (1985).

\bibitem{Montgomery} R.~Montgomery, {\it A Tour of Subriemannian Geometries, Their
           Geodesics and Applications}, AMS, Math. Surveys and Monographs, Vol. 91 (2000).

\bibitem{Sagan} H.~Sagan, {\it Introduction to the calculus of variations}, McGraw--Hill Book Company, New York (1969)

\bibitem{Gracia} Xavier Gracia, Jesus Marin-Solano, Miguel-C. Munoz-Lecanda, {\it Some geometric aspects of variational calculus in constrained systems}, Rep. Math. Phys. 51 Issue 1 (2003), 127-148. 

\end{thebibliography}
\end{document}